\documentclass{amsart}

\usepackage{amssymb}
\newcommand{\C}{\mathbb{C}}
\newcommand{\R}{\mathbb{R}}

\def\be{\begin{equation}}
\def\ee{\end{equation}}
\def\bequnan{\begin{eqnarray*}}
\def\eequnan{\end{eqnarray*}}
\def\la{\label}
\def\von{\varepsilon}

\def\cD{{\mathcal D}}
\def\cP{{\mathcal P}}
\def\P1{{\mathbf P}}
\def\cQ{{\mathcal Q}}

\def\cH{{\sf H}}

\def\cE{{\mathcal E}}

\def\wt{\widetilde}

\def\({\left(}

\def\){\right)}
\def\s2{ \wt{S_2} }
\def\tQ{\widetilde{Q}}
\def\llangle{\left\langle}
\def\rrangle{\right\rangle}
\def\omede{\omega_\delta}
\def\len{\left\|}
\def\rin{\right\|}

\newcounter{prop}
\newtheorem{theorem}{Theorem}
\newtheorem{proposition}[prop]{Proposition}
\newtheorem{corollary}[prop]{Corollary}
\newtheorem{lemma}[prop]{Lemma}

\theoremstyle{definition}
\newtheorem{remark}{Remark}

\begin{document}

\title[Spectral singularities]{Estimates of solutions of linear Boltzmann equation at large time and spectral singularities}

\author[Romanov]{ROMAN ROMANOV} 
\address{Laboratory of
Quantum Networks, Institute for Physics\\ St.Petersburg
State University, 198504, Saint Petersburg, Russia \\ \and \\
Department of Mathematical Physics, Faculty of
Physics\\ St.Petersburg State University, 198504, Saint Petersburg,
Russia}
\email{morovom@gmail.com}

\begin{abstract}
The spectral analysis of the dissipative
linear transport (Boltzmann) operator with polynomial collision integral by the Sz\H{o}kefalvi-Nagy - Foia\c{s} functional
model is given. An exact estimate for the reminder in the asymptotic of the corresponding evolution semigroup is proved in the isotropic case. In the general case, it is shown that the operator has finitely many eigenvalues and spectral singularities and an absolutely continuous essential spectrum, and an upper estimate for the reminder is established. \end{abstract}


\maketitle

\section{Introduction}

This paper is devoted to study of the large time asymptotics for the linear Boltzmann (transport) equation by methods of the Sz\H{o}kefalvi-Nagy -- Foia\c{s} \cite{Na} functional model. After elimination of a constant absorption term the equation takes the form 
\begin{eqnarray}\la{tr} \frac\partial{\partial t} u_t ( x , \mu ) = - \mu \frac\partial{\partial x} u_t ( x , \mu ) + c ( x ) \int_{ -1 }^1 K ( \mu , \mu^\prime ) u_t ( x , \mu^\prime ) \mathrm{d} \mu^\prime , \\ x \in \R , \mu \in [ -1 , 1 ] .  \nonumber \end{eqnarray}

The notation is explained in Preliminaries. This equation describes, for instance, the neutron transport in a slab of multiplicative medium under appropriate simplifying assumptions. The problem involves two parameters: the local mean number of the
secondary particles per collision $c$, a nonnegative compactly supported function on the real line, and the collision operator $ K \in {\mathbf B} L^2 ( -1 , 1 ) $ which describes the angle distributions of the secondaries. It was first considered in \cite{JLh,LW,LW2} in the case of the isotropic distribution, which corresponds to the kernel $ K  \mu , \mu^\prime ) \equiv {\mbox{const}} $. The result of Lehner and Wing says there exists a finite set of $ \beta_j > 0 $ and finite rank projections $ P_j $ such that for any $ u_0 \in L^2 ( \R \times [ -1 , 1 ] ) $ for all $ \delta > 0 $  \be\la{LWas} u_t = \sum_j e^{ \beta_j t } P_j u_0 + O \left( e^{ \delta t } \right) , \;\; t \to + \infty . \ee

The main problem now is to analyze the reminder. Lehner and Wing \cite{LW2} proved that the reminder decays pointwise for $ x $ from the support of $ c $ for $ u $ from a set of initial data dense in a subspace and erroneously claimed a resolvent estimate \cite[Lemma 6]{LW2} which implies that the reminder is $ \len u \rin O ( \ln t ) $. It turned out \cite{KNR} that the estimate does not hold in general. To the best of our knowledge, no results on precise estimates of the reminder in the $ L^2 $--norm have appeared since then, and the later work in the field dealt with other types of transport operator. In the present paper we analyze the structure of the reminder and, in particular, give power upper estimates in the case of a polynomial collision integral, and precise estimates in the isotropic case. 

In terms of the generator, the problem is about the structure of its essential spectrum. Our first main result (Theorem 1) can be stated as follows (the formulation of the theorem in the body of the paper is slightly more detailed).

\medskip

\textbf{Theorem.} {\it Let $ L $ be the operator in $ H = L^2 ( \R \times [ -1 , 1 ] ) $ corresponding to the equation (\ref{tr}) in the sense that $ e^{ itL } u_0 = u_t $ for $ u_0 \in H $. Assume that the kernel $ K $ is polynomial, and the operator $ K \ge 0 $. Then $ L $ is similar to an orthogonal sum of three operators, $ L_1 $, $ L_2 $ and $ L_d $, such that

1. $ L_d $ is a finite rank operator;

2. $ L_2 $ is an absolutely continuous (a. c.) operator with spectrum of finite multiplicity;

3. $ L_1 $ is the orthogonal sum of infinitely many copies of the operator of multiplication by the independent variable in $ L^2 ( \R ) $. 

The operator $ L $ has at most finitely many spectral singularities. All the singularities are of at most finite power order.}

\medskip

\textbf{Corollary.} {\it $ \mathrm{(i)} $ There exist finite $ l , n $ such that the group $ e^{ itL } $ satisfies \be\la{ta} e^{ itL } =  \sum_{ i = 1 }^l e^{ -i\lambda_j t } P_j  + O \left( t^n \right) , \;\; t \to + \infty . \ee Here $ \lambda_j \in \C_+ $, the $ O $ refers to the operator norm, and $ P_j $, $ j \le l $, are finite rank operators. 

$ \mathrm{(ii)} $ $ \sup_{ t > 0 } \len e^{ itL } u \rin $ is
finite for any $ u $ from a dense set in $ \cap_j \ker P_j $}

In the previous paper \cite{KNR} we established assertions 2 and 3 of the Theorem in the isotropic case by a different method. The finiteness of the discrete spectrum in the anisotropic case also appears to be a new result. The argument used for that in \cite{JLh, LW} exploits some sign definiteness property of the bordered resolvent discovered by Lehner and Wing, which is specific for the isotropic problem, and thus cannot be applied in the general case. 

To explain the assertions of the Theorem, recall \cite{Pav} that the invariant subspace of a dissipative operator corresponding to the essential spectrum in general is a sum of the a. c. subspace and an invariant subspace, $ H_s $, corresponding to a sort of singular spectrum. The simplest example of the subspace $ H_s $ being non-trivial is given by the Volterra operator. Assertion 2 says that in the situation under consideration $ H_s $ is trivial. This implies claim (ii) in the Corollary by an abstract theorem. The finite multiplicity of $ L_2 $ means that the linear set of data for which the reminder in (\ref{ta}) may actually grow is, in a sense, thin. The notion of spectral singularity comes from the Sz\H{o}kefalvi-Nagy -- Foia\c{s} criterion \cite{Na}, according to which a dissipative operator is similar to a self-adjoint one if, and only if, its resolvent, $ R ( z ) $, in the upper half plane satisfies $ \| R ( z ) \| \le C \( \Im z \)^{ -1 } $. By definition, spectral singularities are those points on the real axis at which the resolvent of the absolutely continuous component of the operator grows faster than $ \( \Im z \)^{ -1 } $. In applications, spectral singularities were first discovered and studied for the Schr\"{o}dinger operator with a complex potential by Na\u\i mark \cite{Nai} and later analyzed by means of the Nagy -- Foia\c{s} functional model by Pavlov (see \cite{PavEn} and references therein).

Our second main result concerns the isotropic case. In the isotropic case the operator can only have a single spectral singularity located at zero \cite{KNR,LW2}. The singularity does occur. Namely, if we denote by $ \cE $ the set of $c $'s for which there is a spectral singularity, then in \cite{KNR} we showed that for any nonzero compactly supported $ c \in L^\infty $ the function $\varkappa c $, $ \varkappa > 0 $, belongs to $ \cE $ for an infinite discrete set of values of the constant $ \varkappa $. Theorem 2 in the present paper says that this singularity is either logarithmic or of the first order and gives asymptotics of the inverse of the characteristic function of the operator at the singularity. It implies the following assertion (Corollary \ref{semigroest} in the main text). 

\medskip
 
\textbf{Proposition.} {\it In the case of isotropic scattering, let $ c \in \cE $, and let \[ Z_t = e^{ itL } -  \sum_{ i = 1 }^l e^{ \beta_j t } P_j \] in the notation of (\ref{LWas}). Then \[ \len Z_t \rin  \le C ( 1 + t ) \] for all $ t > 0 $, and one of the following alternatives takes place.
 
$ \mathrm{(i)} $ For any sufficiently small $
\von > 0 $ there exists a $ u \in H $ such that
\[ \len Z_t u \rin = t^{ 1 - \von } ( 1 + o
( 1 ) ), \hspace{.5cm} t \to + \infty . \]

$ \mathrm{(ii)} $ The estimate $$ \len Z_t 
\rin \le C \ln t $$ holds for all $ t \ge 2 $, and is exact in the sense that for any
sufficiently small $ \von > 0 $ there exists a $ u \in H $
such that $$ \len Z_t u \rin = \( \ln t \)^{ 1 - \von } ( 1
+ o ( 1 ) ), \hspace{.5cm} t \to + \infty . $$

The alternative $ (i) $ takes place if, and only if, there is an eigenfunction of the integral operator in $ L^2 ( \R ) $ with the kernel $ \frac 12 \sqrt {c ( x )} \ln | x - y | \sqrt {c ( y )} $ with the eigenvalue $ -1 $ orthogonal to the vector $ \sqrt c $.}
 
The proof of Theorem 1 is based on the analysis of the characteristic function of the operator. The characteristic function, $ S $, of a dissipative operator $ L $ is a
contractive analytic operator function in the upper half plane
defined in terms of the resolvent of $ L^* $. In the problem under consideration, the characteristic function is analytic one the real axis except at point $ z = 0 $. This implies  that all non-zero spectral singularities correspond to poles of the characteristic function. The main problem is to analyze the behaviour of the characteristic function at $ z = 0 $.  It turns out that although $ z = 0 $ is not an isolated singularity, $ S^{ -1 } ( z ) $ admits power estimate at it. These assertions combined with the asymptotics of $ S ( z ) $ at infinity imply the absolute continuity of the spectrum. The splitting of the
absolutely continuous component is obtained by application of an
abstract construction of invariant subspaces of
operators with absolutely continuous spectrum suggested in \cite{KNR}. 
The proof of Theorem 2 comes from analysis of the remainder in the asymptotics of the characteristic function.

The structure of the paper is the following. In section 2 we give
a brief description of the abstract construction of separation of the absolutely continuous subspace from \cite{KNR}. Sections 3 and 4 are devoted to proofs of Theorems 1 and 2, respectively. The estimate for the semigroup resulting from Theorem 2 is given in corollary \ref{semigroest}.

\subsection{Notation and Preliminaries}

The following notation is used throughout the paper.

$ \| \cdot \|_2 $ -- the Hilbert-Schmidt norm of operators.

$\C_\pm = \{ \pm z \colon \Im z > 0 \} $; $\omede (z) = \{ z^\prime \in \C_+ \colon |z^\prime - z | \le \delta \} $,

$( X, Y ) $ -- the angle between subspaces $ X , Y $ of a Hilbert space.

For a closed operator $A$ on a Hilbert space $H$

$\sigma_+ ( A ) =
\sigma (A) \bigcap \C_+$,

$ \sigma_{ess} ( A ) $ -- the essential spectrum of $ A $; by definition, $ \sigma_{ess} ( A ) $ is the complement in $ \sigma ( A ) $ of the set of isolated points $ z \in \sigma ( A ) $ such that the corresponding Riesz projection is a finite rank operator;

A subspace
$ \cH \subset H $ is called an \textit{invariant subspace of} $A$ if $ \overline{( A
-\lambda )^{-1} \cH } = \cH $ for all $\lambda \in \rho (L)$. With
that definition, if $ \cH $ is an invariant subspace and $ A $ is an operator with a bounded imaginary part, then $ A \(\cD ( A ) \bigcap \cH \) \subset \cH $ \cite{KNR}, hence the restriction $ A_\cH $ of $ A $ to $ \cH $ with the
domain $ \cD (L) \bigcap \cH $ is a closed operator in $ \cH $;

$ H_{ess} ( A ) \colon = \bigcap \ker \cP_d
$ where $ \cP_d $ ranges over the Riesz projections for points of $ \sigma_+ ( A ) $. The subspace $ H_{ess} ( A ) $ is an invariant subspace of $ A $. We write $ H_{ess} $ for $ H_{ess} ( A ) $ when it is clear which operator the notation refers to. The same convention applies to the subspace $ H_{ac} ( A ) $ defined below; 

A subspace $ {\mathcal J} \subset H
$ is called generating if $ H = \bigvee_{ \lambda \in \rho
( A ) } \( A - \lambda \)^{-1}{\mathcal J}$. The \textit{
multiplicity (of the spectrum)} of the operator $ A $ is the number
$m ( A ) = \inf \dim \mathcal{N} $, where $\mathcal N $ ranges over the
generating subspaces of $A$. 

Given a Hlbert space $ E $, $ H^2_\pm (E) $ stand for the Hardy
classes of $E$-valued functions $f$ analytic in $ \C_\pm $,
respectively, and satisfying $\sup_{ \von
> 0 } \int_\R \len f ( k \pm i\von ) \rin^2_E dk < \infty $. The classes $ H^2_\pm (E) $ are naturally identified with subspaces in $L^2( \R , E)$ comprised by the boundary values of their elements on the real axis. 

Let $L_0 $ be a selfadjoint operator, $V \ge 0 $ a bounded operator, $ L =
L_0 +i V $. For $ z \in \C_+ $ define the operator
\be\la{Qdef} Q ( z ) = i \sqrt
V \( L_0 - z \)^{-1} \sqrt V  . \ee It satisfies $ \Re Q ( z ) \le 0 $. 
A version of the Weyl theorem on relatively compact perturbations holds.

\medskip 

\textbf{Lemma. }(Weyl Theorem) \textit{If $ Q ( z ) $ is a compact operator (at
least at one point $ z \in \C_+ $ and then at all points) then $
\sigma_{ess} ( L ) = \sigma_{ess} ( L_0 ) $, $\sigma_+ ( L ) = \{ z \in \C_+ : \ker ( I + Q ( z ) ) \ne \{ 0 \} \} $.}

Throughout, we use the same notation, $ Q ( z ) $, for the restriction of the operator $ Q (z ) $ to the subspace $ \overline{\mbox{Ran} \, V} $.

Let $ \Phi : \C_+ \to {\bf B} ( E ) $, $E$ being a Hilbert space, be
a bounded analytic operator - function. A scalar function $ m ( z
) \not \equiv 0 $ in $ \C_+ $ is called {\it a scalar multiple}
for $\Phi $ if there exists a bounded analytic operator - function
$\Omega ( z ) $ in $ \C_+ $ such that $ m ( z ) I = \Phi ( z ) \Omega
( z ) = \Omega ( z ) \Phi ( z ) $ for all $ z \in \C_+ $ \cite{N}. A bounded analytic function $ \Phi : \C_+ \to {\bf B} ( E ) $ is
called {\it outer} if $ \overline { \Phi H^2_+ ( E ) } = H^2_+ (
E ) $. Any contractive analytic function $ \Phi : \C_+ \to {\bf B} ( E ) $ admits the {\it canonical
factorization} in a product of two contractive analytic $ {\bf B}
( E ) $ - valued functions of the form $ \Phi = \Phi_i \Phi_e $ where $ \Phi_e
$ is an outer function, and $ \Phi_i ( k ) $ is isometric for a.e. $
k \in \R $. 

\section{Absolutely continuous subspace}

Let $L$ be a maximal dissipative operator in a Hilbert space $ H $
with a bounded imaginary part $ V = \Im L $ such that $
\sigma_{ess} (L) \subset \R $. The operator $ L $ can be represented as an orthogonal sum of a selfadjoint operator and a completely non-selfadjoint one \cite{Na}. The operators in this sum are called the {\it selfadjoint} and {\it completely non-selfadjoint} parts of $ L $, respectively. 
 
\textit{The
absolutely continuous subspace $H_{ac} ( L ) \subset H $} of the
operator $L $ \cite{Pav,Sachn,Tikhang} is the closure of
the linear set $\wt{H_{ac}^w} $ of weak smooth vectors: \begin{eqnarray*} \label{defpla} & H_{ac} ( L ) \stackrel{ \mbox\small{def}}{ =
} \mathop{Clos} \wt{H_{ac}^w} , & \\ \nonumber &
\wt{H_{ac}^w} \stackrel{ \mbox\small{def}}{ = } \left\{
\begin{array}{cl} u \in H \colon & \( L - z \)^{ -1 } u \mbox{ is analytic in }
\C \setminus \R ,
\\ & \llangle \( L - z \)^{ -1 } u , v \rrangle_
\pm \in H^2_\pm \mbox{ for all } v \in H . \\
\end{array} \right\}.
\end{eqnarray*} 

We call the operator $L$ \textit{absolutely
continuous (a. c.)} if $ H = H_{ac} ( L ) $. Motivations and analysis of the definition of the a. c. subspace from various viewpoints including scattering theory can be found in \cite{N,Pavli,PavEn,Nab,Sachn,Tikhang}. One of them is that $ H_{ac} (L) $ is the minimal invariant subspace of $ L $ containing all the
invariant subspaces $ X $ of it such that the restriction $ \left.
L \right|_X $ is similar to an a. c. selfadjoint operator $ A_X $. If $ L $ is completely non-selfadjoint then the subspace $ H_{ac} ( L ) $ coincides with the invariant subspace of $L$ corresponding to the canonical factorization of its characteristic function (to be defined shortly) in the sense that the characteristic function of $ \left. L \right|_{ H_{ac} } $ coincides with the pure part of the outer factor of $ S $ (see \cite{N,equival}). Notice that, obviously, $ H_{ac} ( L ) \subset H_{ess} ( L ) $. 

Let $ E $ stand for the subspace $ \overline{{\rm Ran }V} \subset H $. The characteristic
function $ S(z): \; E\to E $, $ z\in \C_+$, of the operator $L$ is
defined by the formula $$ S(z) = I + 2 i \sqrt V \( L^* - z
\)^{-1} \sqrt{V}. $$ This function is analytic and contractive in
$ \C_+ $. It has nontangential boundary values, $ S(k) \equiv S (
k + i 0 ) $, on the real axis in the strong sense for a.e. $ k\in
\R $. For $ z \in \C_+
\bigcap \rho (L)$ the operator $ S ( z ) $ is boundedly invertible
on $ E $. It can easily be derived by direct calculation that  
\be\la{Char} S ( z ) = \frac{I + Q(z)}{I - Q ( z ) },
\ \ z \in \C_+ . \ee

In the next section we will need the following result to show that the operator corresponding to (\ref{tr}) satisfies $ H_{ess} = H_{ac} $.

\begin{lemma}
\label{abscon} \cite{Na,Pav} Let the selfdjoint part of the operator $ L $ be absolutely continuous. If the characteristic function $ S $
admits a scalar multiple, $ d $, of the form $ d = b d_e $ with $ d_e $ outer and $ b $ a
Blaschke product, then the invariant subspace $ H_d $ of the
operator $ L $ spanned by its root vectors corresponding to
nonreal eigenvalues is complementary to $ H_{ac} ( L ) $, $ \overline {H_{ac} ( L ) \dot{+} H_d } = H $.
\end{lemma}

Let $ S = S_i S_e $ be the canonical factorization of the characteristic function.

A point $ k \in \R $ is called a {\it proper point of the operator
L} if $$\sup_{ z \in \omede ( k ) } \len S_e^{-1} ( z ) \rin <
\infty \, \mbox{ for some  } \, \delta > 0 . $$ We call a point $
k \in \R $ a {\it spectral singularity} if is not proper. This
definition, first suggested in \cite{Pav}, can be shown to be
equivalent to the one described in Introduction, see
\cite{Na,Pav}.

\medskip

\textbf{Proposition (Nagy - Foia\c{s} criterion \cite{Na})}. \textit{The operator $ \left.
L \right|_{ H_{ac} } $ is similar to a self-adjoint operator if and
only if $ \mbox{ess sup}_{k \in \R } \len S^{-1} ( k ) \rin <
\infty $.}

This criterion implies that $ \left. L \right|_{ H_{ac} } $ is not
similar to a self-adjoint operator if it has spectral
singularities, since for an outer function, $ \Theta $, the function $ \Theta^{-1} ( z ) $ is bounded in $ \C_+ $ if, and only if, $ \mbox{ess sup}_{k \in \R } \len \Theta^{-1} ( k ) \rin $ is finite.

A point $ k_0 \in \R $ is said to be a {\it spectral singularity
of at most power order} if there exists a $ p > 0 $ such that \be\la{Sp} \len S^{ -1 } ( k ) \rin \le C \left| k - k_0 \right|^{ -p } \ee for a. e. $ k $ in a vicinity of $ k_0 $ on the
real axis, a {\it spectral singularity
of order $ p > 0 $ in the strict sense}, if (\ref{Sp}) is satisfied, and for some nonzero $ e_0 \in E $ and $ C_1 > 0 $ \[ \len
S ( k ) e_0 \rin_E \le C_1 \left| k - k_0 \right|^p \] for a. e. $ k $.

We now sketch the construction of splitting of $ H_{ac} $ into a sum of invariant subspaces of  $ L $ from \cite{KNR}. The subspaces obtained make nonzero angle and have the property that the restriction of the operator to one of them is similar to a self-adjoint operator while the other contains all the information about the spectral singularities. The result is formulated in Lemma \ref{abstrsep}. The detailed proofs can be found in \cite{KNR}. The reader not interested in the construction may prefer to take Lemma \ref{abstrsep} for a granted existence result, and proceed to the next section. 

Consider the function $ \Delta ( k ) = I - S^* ( k ) S ( k )
$. Define the subspace ${\mathcal X} \subset L^2 (\R; E ) $ to be the closure of $ \Delta L^2 (\R; E ) $. Let $A_0$
be the operator of multiplication by the independent variable in
$ {\mathcal X } $.

\begin{proposition}\cite[Theorem 4, Corollary 1]{N} \cite{Pavli} \cite[Theorem 2]{equival} Let $ L^\prime $ be the completely non-selfadjoint part of $ L $. Then there exists a bounded operator $ W \colon {\mathcal X } \to H $ such that \be \la{inter} \( L - z \)^{-1} W \ =\ W \(
A_0 - z \)^{-1} , \; z \in \rho ( L ) , \ee $ \overline{\mathrm{Ran} W } = H_{ac} ( L^\prime ) $, and for any $ g \in {\mathcal X} $ \be\la{PK} \| S 
g \| \le \len W g \rin_H \le \| g \| , \ee the norms of $ S g $ and $ g $ being in $ L^2 (\R, E ) $.
\end{proposition}

For a reader familiar with the functional model we would like to notice that in a spectral representation of the selfadjoint dilation of $ L^\prime $  
the operator $ W $ essentially coincides with the restriction of the projection on orthogonal complement of incoming and outgoing subspaces to the residual part of the dilation. For others, we just say that the operator $ W $ is constructed, in a sense, explicitly from the operator $ L $ via its characteristic function. 

Define $D ( k ) = S^* ( k ) S ( k ) $, $ k \in \R $. Given a $ \beta \in ( 0, 1 ) $ let $ X_1 ( k ) $, $ X_2 ( k )
$ be the spectral subspaces of the self-adjoint operator $ D ( k )
$ for the intervals $ [ 0, \beta^2 )$, $[\beta^2 , 1 ) $,
respectively. Then, the subspaces $ X_{1,2} = \{ f \in L^2 ( \R, E ):\; f \in
X_{1,2} ( k ) \mbox { for a.e. } k \in \R \} \subset {\mathcal X}  $ reduce the operator $ A_0 $, $ {\mathcal X} = X_1 \oplus X_2 $. By the intertwining relation (\ref{inter}), $\cH_{1,2} = \overline{W X_{1,2}}$ are invariant subspaces of the operator
$L$. As is clear from (\ref{PK}), the operator $ L_2^\beta = \left. L \right|_{ \cH_2^\beta }$ is similar to the restriction of $A_0 $ to $ X_2 $ because the restriction of $ W $ to $ X_2 $ is boundedly invertible, $ \len W u \rin \ge \beta \len u \rin $ for all $ u \in X_2 $.  Let us estimate the angle between $ \cH_1 $ and $ \cH_2 $. For any nonzero $ u \in X_2 $, $ v \in X_1 $ on account of (\ref{PK}) we have \[ \| W u \| \le \| u \| \le \frac 1\beta \len S u \rin \le  \frac 1\beta \len S ( u - v ) \rin \le \frac 1\beta \len W ( u - v ) \rin \le \frac 1\beta \len W u \rin \sin ( \cH_1^\beta, \cH_2^\beta ) . \] Here we have used the fact that $ S u \perp S v $ for our choice of subspaces $ X_{ 1,2 } $. The last inequality in the chain is just the definition of the angle between subspaces of a Hilbert space. Thus, $ \sin ( \cH_1^\beta, \cH_2^\beta ) \ge \beta $. 

\begin{lemma}\cite[Corollary 2.8]{KNR}\la{abstrsep} Suppose that $ S $ is norm continuous
in the closed upper half plane, $ S ( k ) - I \in {\bf S}^\infty $
for all $ k \in \R $, and $ S ( k ) \to I $ in the operator norm
as $ |k| \to \infty $. Then $ H_{ac} $ can be represented as a
linear sum $ H_{ac} = \cH_1 \dot{+} \cH_2 $ of invariant subspaces
$ \cH_{ 1,2 } $ of the operator $ L $ such that

1. The multiplicity of spectrum of the restriction $ L_1 = \left. L \right|_{
\cH_1 } $ is finite and coincides with the number $ \max_{ k \in \R } \dim
\ker S ( k ) $,

2. $ L_2 = \left. L \right|_{ \cH_2 } $ is similar to the a.c. selfadjoint operator of multiplication by the independent variable restricted to the subspace $ X_2 \subset L^2 ( \R , E ) $ in the construction above.

3. $( \cH_1, \cH_2 ) > 0$.
\end{lemma}

\begin{proof} Let $ d ( k ) $ be the least positive eigenvalue of $ D ( k ) $. The assumptions of continuity and compactness imply that the dimension $ \dim \ker S( k ) $ is bounded above in $ k $, the maiximum of $ \dim \ker S( k ) $ is attained, and $ \inf_{ k \in \R } d ( k ) > 0 $. Pick a positive $ \beta $ such that $ \beta^2 \le \inf_{ k \in \R } d ( k ) / 2 $. The intertwining relation (\ref{inter}) implies that the multiplicity $ m \( L_1^\beta \) $ of the restriction $ L_1^\beta = \left. L \right|_{ \cH_1^\beta }$ is not greater than the maximum of $ \dim \ker S( k ) $. It is easy to see \cite{KNR} that under the imposed assumptions $ \ker W = \{ 0 \} $, and therefore $ m \( L_1^\beta \) $, in fact, equals to the maximum of $ \dim \ker S( k ) $. Thus, the subspaces $ \cH_{ 1,2 }^\beta $ possess the properties 1 through 3. \end{proof}

\section{Transport Operator for Slab: Anysotropic Case}

The phase space of the transport problem for a slab is
$\Gamma = \R \times \Omega ,$ $\Omega \equiv [-1,1]$. The
variables $x\in \R $ and $ \mu \in \Omega $ make sense of the
position and the cosine of the angle between the momentum and the
coordinate axis, respectively. The densities of particles are elements of $ H = L^2(\R \times \Omega )$. Let $ c \in L^\infty (\R ) $ be an a. e.
nonnegative compactly supported function, and $K $ be an integral operator in $ L^2(\Omega )$ with a degenerate nonnegative polynomial kernel, that is, of the form \[ K = \sum_{ i = 1 }^N k_i^2 \llangle \cdot ,  {\mathcal P}_i \rrangle_{  L^2(\Omega ) } {\mathcal P}_i \] where $ {\mathcal P}_i $ are polynomials of unit norm in $ L^2 ( \Omega ) $, and $ k_i  > 0 $ for all $ i \le N $. It is supposed throughout that the constant function is an eigenfunction  of $ K $. The physical meaning of parameters $ c $ and $ K $ is explained in Introduction. The case of isotropic scattering corresponds to the operator $ K $ of the form $ K = \frac 12 \int_\Omega \cdot
\ d \mu^\prime $ ($ N =1 $), where $ d \mu $ stands for the Lebesgue measure
on $ \Omega $.

We do not distinguish between the operator $ K $ and the
operator $ I \otimes K $ in $ L^2(\R\times \Omega ) = L^2(\R )
\otimes L^2( \Omega )$ in our notation. Using the basis of polynomials $ \{ {\mathcal P}_i \} $ one can naturally identify the
range of $ K $ in $ H $ with the space $L^2 (\R , \C^N ) $ of vector-functions
of the variable $x$, $K H = L^2( \R )\otimes {\mathcal L} \{ {\mathcal P}_j \}_1^N \simeq L^2( \R , \C^N ) $.

We assume that the units are chosen so that the absorption cross-section equals to $ 1 $. Let $ v_t \in H $ be the density of particles at the moment $t$, and let $ u_t = v_t e^t $. Then the evolution of the effective density $ u_t $ in the space $ H $ is given by the
solution of the Cauchy problem for the Boltzmann equation (\ref{tr}). Define the one-speed transport operator $L$ in the space $ H $ by the following expression, \be\la{L} L = i\mu \partial_x - i
c(x) K , \ee on the domain $\cD = \{ f \in H:\, (1) f(\cdot ,\mu
)\,\mbox { is absolutely continuous for a.e. } \mu \in \Omega ;\;
(2) $ $ \mu \partial _x f \in H \}$ of the self-adjoint operator $
L_0 = i\mu \partial _x $. The imaginary part of $L$ is obviously bounded, and the group $ e^{ it L } $ is defined for all $ t \in \R $. With this notation, we have \[ u_t = e^{ it L } u_0 . \]
Instead of $L$, it is convenient to deal with the
dissipative operator $ T = L^* $. The spectral analysis of operators $ L $ is reduced to that of operators $ T $ since $ L $ is unitarily equivalent to the operator $ -\wt T $ where $ \wt T = i\mu \partial_x + i
c(x) \wt K $ with $ \wt K = J K J^* $, $ ( J f ) ( x, \mu ) = f (
x, - \mu ) $: we have $ L = J ( - \wt T ) J^* $. 

Let $ V = \Im T $. The subspace $E \equiv \overline{ \mbox
{Ran}\, V} $ is then naturally identified with a subspace in $ L^2 ( \R , \C^N ) $ via the basis of eigenfunctions of the operator $ K $. Let $Q ( z ) =
\left.  i\sqrt V \( L_0 - z \)^{-1} \sqrt V \right|_E $, $ z \in\C_+ $. The resolvent of  $ L_0 $ for $ \Im z > 0 $ has the form \[ \( \( L_0 - z \)^{-1}  f \) ( x , \mu ) = \frac i\mu \left\{ \begin{array}{cc}
{\displaystyle - \int_{ - \infty }^x e^{ - i z( x - s )/\mu } f ( s , \mu ) \mathrm{d} s } , & \mu < 0 \\
\noalign{\vskip8pt} {\displaystyle \int_x^\infty e^{ - i z( x - s )/\mu } f ( s , \mu ) \mathrm{d} s }, & \mu > 0 .\end{array} \right. \]
Plugging this into the definition of $ Q ( z ) $, we find after a change of variable that \be\la{Qany} Q ( z ) = \sum_0^{ 2N } T_j ( z ) \otimes G_j \ee where $ T_j ( z ) $ are integral operators in the space $ L^2 ( \R ) $ of scalar functions with the kernels $ t_j ( x , y ) =  \sqrt{c ( x ) } \( \mbox{sign} ( x - y )\)^j  \, E_j ( - i z | x -
y | ) \sqrt{c ( y ) } $, respectively, restricted to the closures of their ranges, \be\la{Ej} E_j ( s ) = \int_1^\infty e^{-st} t^{ -j-1 } \mathrm{d}t \ee for $\Re s > 0 $, and $ G_j $ are $ N \times N $-matrices with real entries. Notice that the matrix $ G_0 \ne 0 $ by the assumption that the constant is an eigenfunction of the operator $ K $. Since
$c$ is compactly supported, the operator $Q(z)$ is of the
Hilbert - Schmidt class for all $z\in \C_+$. In the following lemma we gather the required properties of the operator $ Q ( z ) $. The superscript $ 0 $ is omitted throughout for $ j = 0 $ so $ E $ stands for the function $ E_0 $. 


\begin{lemma}\la{propertiesQany} 
 $ \mathrm{(i)} $ The function $ Q ( z ) $ admits the
representation \be\la{assyanisotro} Q ( z ) = \ln ( -i z ) \llangle \cdot , \ell \rrangle \ell + B + \Theta
(z) \ee where $ B $ is a Hilbert--Schmidt operator, $ \ell \in E $ is given by \[ \ell = \sqrt c {\mathcal P} , \; {\mathcal P} \stackrel{ \mbox\small{def}}{ =
} \( \begin{array}{c} k_1 {\mathcal P}_1 ( 0 ) \\ k_2 {\mathcal P}_2( 0 ) \\ \cdots \\ k_N {\mathcal P}_N ( 0 ) \end{array} \) \in \C_N , \] and $ \Theta ( z ) $ is an analytic function in ${\mathcal O } \equiv \C
\setminus \{ -it , t > 0 \} $ such that $ \| \Theta ( z ) \| \le C | z \ln z | $ for $ |z| \le 1/2 $. This formula defines an analytical continuation of the function $ Q ( z ) $ from $ \C_+ $ to ${\mathcal O } $.

$ \mathrm{(ii)} $ $\len Q ( z ) \rin \to 0 $ as $
|z| \to \infty $ uniformly in $ \mbox{arg} \, z \in [ 0 , \pi ] $.
\end{lemma}

\begin{proof} Integrating by parts, we obtain that $E_j ( s ) = p_{ j - 1 } ( s ) e^{ -s } + c_j s^j E ( s ) $ where $ p_{j-1 } $ is a polynomial of degree $ j - 1 $ and $ c_j = \( -1 \)^j /j! $. The function $E(s)$ admits the representation \cite{LW} \be\label{assy} E(s) = -\ln s - \gamma +\theta (s) , \; \Re s > 0
\ee where $\theta (s) = - \sum_{m=1}^\infty \frac{(-s)^m}{m!m} $
is an entire function and $\gamma $ is the Euler constant. This formula defines an analytical continuation of the function $ E $, and hence of all $ E_j $'s, to $ \C \setminus \{ t \in \R : t < 0 \} $.
Plugging these representations for $ E_j $ into (\ref{Qany}), re-grouping terms according to their behaviour in the $ z $ variable, and taking into account that $ G_0 = \{ k_i k_j {\mathcal P}_i ( 0 ) {\mathcal P}_j (0) \}_{ i, j = 1 }^N $, we obtain the assertion $ \mathrm{(i)} $. To establish the assertion $ \mathrm{(ii)} $, we are going to show that $\len T_j ( z ) \rin_2 \to 0 $ as $
|z| \to \infty $ uniformly in $ \mbox{arg} \, z $ for all $ j $. First, we once integrate by parts in (\ref{Ej}) (in the direction opposite to the one used in  $ \mathrm{(i)} $) to find out that there exists a $ C > 0 $ such that $ | E ( s ) | \le C \left| s \right|^{ -1 } $ for all $ s $ with $ \Re s \ge 0 $. Let $ \chi_\von $ be the indicator function of the interval $ [ 0 , \von ) \subset \R $. Split the operator $ T_j $ into the sum $ T_j = T_j^\von + R_j^\von  $ where $ T_j^\von $ is the integral operator with the kernel $ t_j ( x , y ) \chi_\von ( | x - y | ) $. The estimate on $ E $ we have just found implies that $ \len R_j^\von \rin_2 \le C \left| z\right|^{ -1 } \von^{ -1 } $. On the other hand, by direct inspection of (\ref{Ej}) and (\ref{assy}) we notice that $ \len T_j^\von \rin_2 \le C \von^{ 1/2 } $ for $ j \ge 1 $, and $ \len T_0^\von \rin_2 \le C \von^{ 1/2 } ( 1 + | \ln ( \von |z| ) )| $. Now setting, for instance, $ \von = \left| z \right|^{ - 1/2 } $, we obtain $ \mathrm{(ii)} $.
\end{proof}

Let $ S $ be
the characteristic function of the operator $T$, and define $ \tQ ( z ) $ to be the sum of the first two terms in the right hand side of (\ref{assyanisotro}), so \be\la{tQ} Q ( z ) =  \tQ ( z ) + \Theta ( z ) , \; z \in \C_+ . \ee Fix an arbitrary $ \delta > 0 $ such that $ I - \tQ ( i\delta ) $ is boundedly invertible. Such a $ \delta $ exists since $ \Re Q ( z) \le 0 $ for all $ z \in \C_+ $, and $ \| \Theta ( z ) \| \to 0 $ when $ z \to 0 $.

\begin{proposition}
\la{S0any} The function $ S $ is
analytic on the real axis except at the point $ 0 $, continuous in $ \overline{
\C_+ } $ in the operator norm, and admits analytic continuation from $ \C_+ $ to the set $ \Pi = \{ z \in \C \colon \mbox{arg z } \ne -  \pi /2 , \; | z | \le \delta_0 \} $ for some $ \delta_0 > 0 $. The following asymptotics holds for this analytic continuation in the operator norm uniformly in $ \mbox {arg} \, z $,  
\begin{eqnarray} S ( z ) = S ( 0 ) + \frac 2{ \vartheta_c (1 + \alpha (z)
\vartheta_c )  } \llangle \cdot , \( I - {\tQ}^* (i\delta) \)^{-1} \ell \rrangle \( I -
\tQ (i\delta) \)^{-1} \ell \nonumber \\  + O \( | z \ln z | \) , \; z \in \Pi , \; z \to 0 ,  \label{este0} \\ S(0) = \frac{ I + \tQ ( i \delta ) } { I - \tQ  ( i \delta ) } - \frac 2 {
\vartheta_c } \llangle \cdot , \( I - {\tQ}^* ( i \delta ) \)^{-1} \ell \rrangle \( I - \tQ ( i \delta ) \)^{-1} \ell ,  \label{este} \\ \alpha (z) = - \ln \( - \frac{ i z}\delta \) , \;
\vartheta_c  = \llangle \( I - \tQ ( i \delta ) \)^{-1} \ell, \ell \rrangle \nonumber
.    \end{eqnarray} 
\end{proposition}

\begin{proof} From (\ref{Char}) we infer that the function $ S $ is analytic on the real axis off $ z = 0 $ since such is the function $ Q ( z ) $. Let us establish the required 
analytic continuation and find an asymptotic at $ 0 $ for the function $ \( I - Q (z) \)^{-1} $. Then the corresponding result for $ S $ reads from (\ref{Char}) rewritten as \be\la{char1any} S(z) = - I + 2 ( I
- Q (z) )^{-1} . \ee We have $ \tQ
(z) = \tQ ( i \delta ) + \alpha (z) \llangle \cdot , \ell \rrangle \ell $. A straightforward
computation gives \begin{eqnarray} \label{tQany} \( I - \tQ (z) \)^{-1} = \( I
- \tQ ( i\delta ) \)^{-1} - \frac { \alpha (z) } { 1 + \alpha (z)
\vartheta_c }\llangle \cdot , \( I - {\tQ}^* ( i \delta ) \)^{-1} \ell \rrangle \nonumber \\ \( I - \tQ ( i \delta ) \)^{-1} \ell . \end{eqnarray}The constant $ \vartheta_c $ is nonzero, for otherwise the norm of $  \( I - \tQ (z) \)^{-1} $ would be unbounded in a vicinity of $ z = 0 $, which is not the case since $ \Re Q ( z) \le 0 $ for all $ z \in \C_+ $.  The right hand side of this formula is
continuous at $ 0 $ in $ \overline{ \C_+ } $, admits analytic extension to the set $ \Pi $ when $ \delta_0 > 0 $ is small enough, and this extension is bounded in norm when $ z $ ranges over $ \Pi $. Hence, the function $ \( I- \tQ
(z) \)^{-1} \Theta (z) $ admits analytic extension to the set $ \Pi $ and is $ O ( |z \ln z | ) $ in that set by Lemma \ref{propertiesQany}. Now, we express $ \( I - Q ( z ) \)^{ -1 } $ via $ \( I - \tQ ( z ) \)^{ -1 } $ through the resolvent identity, \[ \( I - Q(z) \)^{-1} = \( I - \( I- \tQ
(z) \)^{-1} \Theta (z) \)^{ -1 }  \( I -
\tQ (z) \)^{-1} . \]
It shows that the function $\( I - Q (z) \)^{-1} $ also extends analytically to the set $ \Pi $, and its extension admits the representation \bequnan \( I - Q (z) \)^{-1} = \( I - \tQ ( i\delta ) \)^{-1} - \frac { \alpha (z) } { 1 + \alpha (z)
\vartheta_c } \llangle \cdot , \( I - {\tQ}^* ( i \delta ) \)^{-1} \ell \rrangle \\ \( I - \tQ ( i \delta ) \)^{-1} \ell + O \( | z  \ln z | \) , \; z \in \Pi . \eequnan On account of (\ref{char1any}), this asymptotics implies continuity of $ S $ at zero and the formulae (\ref{este0}),(\ref{este}). \end{proof}

The required $ \delta $ can be estimated explicitly in terms of $ c $ and $ K $. 

\begin{remark} Proposition \ref{S0any} holds with any positive $ \delta $ such that \[ \delta \le  \min \left\{ \frac 1{2a} , C_N \frac 1{ a \len K \rin^2  \left| c\right|_1^2 } \right\} \] where $ |c|_1 = \int_\R c(x) dx $, $ a $ is the diameter of the support of the function $ c $, and $ C_N $ is a number constant depending on $ N $ only. 
\end{remark}

\begin{proof} Notice that for $ 0 < s \le 1 $ from the definition of $ E_j $'s we have $ | E_j ( s ) - E_j ( 0 ) | \le c_j  s^{ 1/2 } $ for $ j \ge 1 $, and $ | \theta ( s ) | \le  s $, and that $ \| G_j \| \le C \| K \| $ for all $ j $ with a constant $ C $ depending on $ N $ only. Hence from (\ref{Qany}) and (\ref{assy}) we have $ \len \Theta ( z )  \rin_2^2 \le C \len K \rin^2 \left| c\right|_1^2 \left| 2 a z \right| $ provided that $ 2a |z | \le 1 $, and the assertion follows. \end{proof} 

\begin{theorem} The essential
spectrum of the operator $T$ coincides with the real axis: $\sigma_{ess}(T) =
\R $. The non-real spectrum $\sigma ( T ) \cap \C_+ $ is discrete and consists of at most finitely many eigenvalues. Let $ H_{ess} $ be the invariant subspace
of $ T $ corresponding to the essential spectrum defined in Preliminaries. Then the restriction $\left. T\right|_{H_{ess}}$ is absolutely
continuous, $ H_{ess} = H_{ac} ( T ) $. The operator $ T $ has at most finitely many spectral singularities. If $ z = 0 $ is a spectral singularity then it is of finite power order. All other spectral singularities are of finite power order in the strict sense. There exist invariant subspaces $ \cH_{ 1 , 2} \subset H_{ess} $ forming a nonzero angle and such that $ H_{ess}  = \cH_1 \dot{+} \cH_2 $ and the restrictions $T_{1,2} = \left. T \right|_{ \cH_{ 1,2} }$ satisfy

1. The spectrum of $T_1 $ has finite multiplicity equal to $ \max_{ k_j \in \sigma_0 } \dim \ker S (k_j) $, where $ \sigma_0 $ is the set of spectral singularities. 

2. $T_2$ is similar to a self-adjoint operator with Lebesgue spectrum of infinite multiplicity on the real line.\end{theorem}

\begin{proof} By the Weyl theorem, the essential spectrum of $ T $ coincides with that of  $ L_0 $, so $ \sigma_{ess} ( T ) = \R $. The analyticity of $ Q ( z ) $ on the real axis off $ z = 0 $ implies that $ \sigma_0 $ and $\sigma_+ ( T ) $ do not accumulate to a real $ k \ne 0 $. Then, by Lemma \ref{abstrsep} there exist subspaces $ \cH_{ 1, 2} \subset H_{ac} $ such that $ H_{ac} = \cH_1 \dot{+} \cH_2 $, property 1 is satisfied and $ T_2 = \left. T \right|_{\cH_2} $ is similar to the restriction of the operator of multiplication by the independent variable to the subspace $ \{ f \in L^2 ( \R , E ) \colon f ( k ) \in X_2 ( k ) \} \subset L^2 ( \R , E ) $, $ X_2 ( k ) $ being the range of spectral projection of the operator $ S^* ( k ) S( k ) $ corresponding to the interval $ [ \beta^2 , 1 ) $. Thus, property 2 will be established if we show that $ \mbox{rank}
\Delta ( k ) = \infty $ for all $ k \in \R $. Indeed, it is
obvious from (\ref{Char}) that \be\la{S}  \| S ( k ) f \| < \| f \| \ee for $
f = ( I - Q ( k ) ) \varphi $ with $ \Re Q ( k ) \varphi \ne 0 $.
The operator $ \Re Q ( k ) $ has infinite rank since its integral kernel for any real $ k $ has a logarithmic singularity
at the diagonal by virtue of the fact that the matrix $ G_0 $ in (\ref{Qany}) is non-zero. Hence (\ref{S}) is satisfied on a subspace of $f $'s of infinite dimension. Property 2 is proved.

Let us show that sets $ \sigma_+ ( T ) $ and $ \sigma_0 $ are finite. We will actually show that the points $ z \in \overline{\C_+} $ such that $ \ker \( I + Q ( z ) \) \ne \{ 0 \} $ do not accumulate at $ 0 $. Throughout the rest of the proof, $ \P1 $ is a rank 1 operator independent of $ z $ exact form of which is not required. Let $ \cP $ be the Riesz projection of the operator $ I + B $, the $ B $ being the one from (\ref{assyanisotro}), corresponding to the point $ 0 $, $ \cQ = I - \cP $. Then $ I + B \cQ $ is boundedly invertible, and  \[  I + Q (z) = ( I + B \cQ ) \( I + B \cP  + \ln ( -iz) \P1 + \wt\Theta ( z ) \) \] where $ \wt\Theta ( z ) $ is an analytic function in $ {\mathcal O } $ admitting the representation \[ \wt\Theta ( z ) = \Theta_1 ( z ) + \ln ( -iz ) \Theta_2 ( z ) \] with entire functions $ \Theta_{ 1,2 } $ such that $ \Theta_1 ( z ) = \Theta_2 ( z ) = 0 $. Thus, $ \ker ( I + Q ( z )  ) $ is nontrivial if, and only if, such is the kernel of \[ I + B \cP  + \ln ( -iz) \P1  + \wt\Theta ( z ) = \( I + \Theta_1 ( z ) \) \biggl( I + H_1 ( z ) \cP + \ln ( -iz) \(  \P1 + H_2 ( z ) \) \biggr) \] where $ H_{ 1 , 2 } ( z ) $ are functions analytic in a vicinity of $ z = 0 $, and $ H_2 ( 0 ) = 0 $. Let $ P $ be  the orthogonal projection on the orthogonal complement of $ \ker \P1 \cap \ker \cP$. Continuing to factor out invertible terms and developing the inverses of $ I + \mathrm{small} $ operators in the Neumann series we have, \[ I + H_1 ( z ) \cP + \ln ( -iz) \(  \P1 + H_2 ( z ) \)  = ( I + \ln ( -iz ) H_2 ( z ) ) \( I + H ( z ) P \) , \] where \[ H ( z ) = G (z ) + \sum_{ j = 1 }^\infty \ln^j ( -iz ) F_j ( z ) \] with functions $ F_j $ and $ G $ analytic in a vicinity of zero, and satisfying $ \len F_j ( z ) \rin \le \( C |z| \)^{ j -1 } $ for $ j \ge 1 $ and $ z $ in that vicinity. Here we took into account that $ \cP P = \cP $, $ \P1 P = \P1 $. Because of its triangle structure, the operator $ I + H ( z ) P $ has a non-trivial kernel if, and only if, the restriction of $ I +  P H ( z ) $ to the range of $ P $ has. Thus, there is a vicinity $ U^\circ $ of zero such that $ \ker ( I + Q ( z ) ) \ne \{ 0 \} $ for $ z \in U^\circ $ if, and only if, $ \det ( I + P H ( z ) P ) = 0 $, and \be\la{PHP} \len \( I + Q ( z ) \)^{ -1 } \rin \le \frac C{\left| \det ( I + P H ( z ) P ) \right| } \( 1 +  \len  P^\perp H ( z ) P \rin  \) \len  P H ( z ) P \rin^{ {\textrm{\scriptsize{rank}}} \, P } \ee where $ C $ depends on $ \mbox{rank } P $ only. On developing the determinant, \[ \det ( I + P H ( z ) P ) = \sum_{ j = 0 }^\infty \ln^j ( -iz ) f_j ( z ) \] with functions $ f_j ( z ) $ analytic in $ U^\circ $ and satisfying $ | f_j ( z )| \le \( C \left| z \right| \)^{j - M } $, $ j > M $, for an $ M $ large enough ($ M > \mbox{rank} P $ will do). Suppose now that there is a sequence $ z_l \to 0 $, $ z_l \in \overline{ \C_+ } $, such that $ \det ( I + P H ( z_l ) P ) \to 0$. Taking the limit $ z_l \to 0 $, we find consecutively that $ f_M ( 0 ) = 0 $, $ f_{ M - 1 }( 0 ) = 0 $, $ \dots , f_0 ( 0 ) = 0 $. Since $ f_j ( 0 ) = 0 $ if $ j > M $, this means that $ f_j ( 0 ) = 0 $ for all $ j \ge 0 $. Consider the function $ g ( z ) = z^{ -1 } \det ( I + P H ( z ) P ) $. The modulus of $ g ( z ) $ is either bounded away from $ 0 $, or there exists a sequence $ z_l^\prime \to 0 $, $ z_l^\prime \in \overline{ \C_+ } $, such that $ g \( z_l^\prime \) \to 0 $. In first case we have, $ \left| \det ( I + P H ( z ) P ) \right| \ge C |z| $. In the second one, arguing as above, we find that $ f^\prime_j ( 0 ) = 0 $ for all $  j \ge 0 $. Then, we consider the function $ z^{ -2 } \det ( I + P H ( z ) P ) $. This function either has the modulus bounded away from zero, or $ f^{\prime\prime}_j ( 0 ) = 0 $ for all $ j $, etc. This process must terminate after finitely many steps, for otherwise we obtain that all of $ f_j $'s, and hence the determinant, vanish identically in the vicinity $ U^\circ $ which contradicts the discreteness of $ \sigma_+ $. Thus, there exists an $ n < \infty $ such that \[ \left| \det ( I + P H ( z) P ) \right| \ge C \left| z \right|^n \] for $ z \in \overline{\C_+} $ in a vicinity of $ 0 $. This implies that $ \sigma_+ ( L ) $ and $ \sigma_0 $ do not accumulate at $ 0 $ and hence are finite. Plugging this inequality into (\ref{PHP}), and taking into account that $ \| H ( z ) \| = O ( | \ln z | ) $  in a vicinity of $ 0 $, we find that \[ \len \( I + Q ( z ) \)^{ -1 } \rin \le \frac C{\left| z \right|^n} \left| \ln^{  {\textrm{\scriptsize rank}} \, P + 1 } z \right|. \] Rewriting (\ref{Char}) as $ S^{ -1 } ( z ) = -I + 2 \( I + Q ( z ) \)^{ -1 } $ we obtain that \[ \len S^{ -1 } ( z ) \rin \le C \left| z \right|^{ -n-1 } , z \in \overline{ \omede ( 0 )}, \, z\ne 0 . \] This means that the spectral singilarity at $ z = 0 $ is of at most power order. 

We will establish the equality $ H_{ess} = H_{ac} $ by showing that the space $ H $ is a closed linear hull of $H_{ac} $ and  $ H_d $. We have just proved that the subspace $ H_d $ is finite dimensional, so the angle $ \( H_{ac} , H_d \) $ is non-zero. We are going to show that $ S ( z ) $ admits scalar multiple. The fact that all singularities of $ S $ are of at most power order means that $ \prod_{ k_j \in \sigma_0 } \( z - k_j \)^{ m_j }  S^{ -1 } ( z )  $ is a bounded function on any compact contained in $ \overline{ \C_+ }\setminus \sigma_+ ( T ) $ for some set of positive constants $ m_j  $. Then, by (ii) of lemma \ref{propertiesQany}, $ S ( z ) \to I $ when $ z \to \infty $ in $ \C_+ $ uniformly in $ \mbox{arg} \, z $. Now, let $ b $ be the Blaschke product corresponding to $ \sigma_+ ( T ) $, that is, $ b ( z ) = \prod_{ z_j \in \sigma_+ ( T ) } \( \frac{ z - z_j }{ z - \overline{z_j}} \)^{ l_j } $, $ l_j $ being the algebraic multiplicity of the eigenvalue $ z_j $. Combining the established properties of $ S ( z ) $, we obtain that the function $ \pi ( z ) = \( z + i \)^{ -J } \prod \( z - k_j \)^{ m_j } b ( z )  $, $ J = \sum_j m_j $, is a scalar multiple for $ S ( z ) $. This function is obviously a Blaschke product times an outer function. It follows from lemma \ref{abscon} that $ H_{ac} \dot{+} H_d = H $, since the selfadjoint part of $ T_{ ess } $ is absolutely continuous (it is a restriction of the a. c. selfadjoint operator $ L_0 $). 
\end{proof}

 \begin{remark} Slightly modifying the construction of Section 2, for any $ \delta $ small enough one can choose the invariant subspaces $ \cH_{ 1,2} $ satisfying the properties 1 and 2 and $ H_{ess} = \cH_1 \dot{+} \cH_2 $ so that $ \sigma ( T_1 ) $ is a finite union of intervals of common length $ \delta $. In the isotropic case in \cite[Lemma 3.11 and Theorem 3.13]{KNR} a quantitative characterization of this decomposition is given. \end{remark}
 
We now establish the Corollary in the Introduction. Claim $ (\mathrm{ii}) $ of it is an abstract fact holding for any maximal dissipative operator with a. c. spectrum\footnote{We would like to take the chance to correct a misprint in formulation of Corollary 3.15 in \cite{KNR}: the signs of the norm around $ e^{ itL }u $ went missing.} \cite{Pav}. Claim $(\mathrm{i})$ is inferred by applying the following assertion \cite[Corollary 2]{NR}. 

{\it Let $ D $ be a maiximal dissipative operator having finitely many spectral singularities, $ k_j $, and let $ \mathcal S $ be its characteristic function. If for some real $ p > 0 $  the quantity $ \left| k - k_j
\right|^p \len {\mathcal S}^{ -1 } ( k ) \rin $ is essentially bounded in a vicinity of $ k_j $ on the real axis for each $ j $, and $ \mbox{ ess sup}_{ | k | > b } \len {\mathcal S}^{ -1 } ( k ) \rin $ is finite for some $ b  $, then there exist a $ C > 0 $ such that $$ \len \left. e^{ -it D } \right|_{ H_{ac} (D) } \rin \le C \( 1 + t^p \) $$ for all $ t > 0 $.}

\section{Singularity in Isotropic Case}

In this section the operator $ K $ in $ L^2 ( \Omega) $ has the form $ K =  \( 1/2 \)\langle \cdot , {\bf 1} \rangle {\bf 1} $ where $ {\bf 1} $ is the indicator of $ [ -1 , 1 ] $. The space $ E $ then consists of functions independent on the $ \mu $ variable and is identified with the subspace of $ L^ 2 ( \R ) $ made of functions vanishing a. e. outside the set $ \{ x \in \R \colon c ( x ) \ne 0 \} $, and $ Q ( z ) $ is an integral operator with the kernel $ -\frac 12\sqrt{c ( x ) }E ( - i z | x -
y | ) \sqrt{c ( y ) }$. The operator $\tQ ( z ) $ in the representation (\ref{tQ}) has the kernel $\frac 12
\sqrt{c ( x ) } ( \ln ( -i z | x - y | ) + \gamma ) \sqrt{c ( y )
}$, and $ \Theta (z)$ is an entire function such that $ \Theta ( 0 ) = 0 $. 

\begin{lemma}\la{propertiesQ}\cite{LW}
The function $ Q ( z ) $ satisfies $ \pm \Im Q ( z ) > 0 $ when $ \mp \Re z > 0 $, $ \Im z \ge 0 $. \end{lemma}

\begin{proof}
We only sketch the proof referring to \cite{LW,KNR} for details. For $ \Im z > 0 $ consider the operator $ \Xi ( z ) = \left.
i K R_0 ( z ) \right|_{ K H} $ acting in the space $K H $. In the
Fourier representation with respect to the variable $x$, this
operator acts as the multiplication by the function \[
\xi ( p , z ) = \frac i2 \int_{-1}^1 \frac{ d\mu }{ p\mu - z }
.\] On calculating the integral, we find that $ \Im \xi ( p , z ) $ is positive (resp. negative)  if $ \Re z $ is negative (resp. positive), and thus $ \pm \Im \Xi ( z ) $ is a positive operator when $ \mp \Re z > 0 $. It then follows that $ \pm \Im Q ( z )  > 0 $ when $ \mp \Re z > 0 $, since $ \Im Q(z) = \left.
X \Im \Xi (z) X \right|_E $, where $ X \colon L^2 ( \R ) \to L^2 ( \R ) $ is the operator of multiplication by the function $ \sqrt c $. The assertion for $ \Im z = 0 $ is verified by the taking the limit $ \Im z \downarrow 0 $ in the expression for $ \Im \xi ( p , z ) $. This limit also turns out to be a strictly positive or negative function, depending on the sign of $ \Re z $, which can be easily shown to imply that $ \pm Q ( k ) > 0 $ when $ \mp k > 0 $. \end{proof}

The lemma implies immediately that the nonreal spectrum $\sigma _+ ( T ) $ of the operator $T$ lies on the imaginary axis \cite{JLh}. Combined with the fact that $ Q ( k ) \to I $ in norm as $ |k| \to \infty $ established by lemma \ref{propertiesQany}, it also implies that 

\begin{corollary} \la{sup} For any real $ k \ne 0 $ the operator $ S ( k ) $ has a bounded inverse, and \be \label{estme}
\sup_{k \in \R \setminus [ -\delta, \delta ] } \len S^{-1} (k)
\rin < \infty \ee  for any $ \delta > 0 $. \end{corollary} 

According to the lemma, $ \pm Q ( z ) $ is a Herglotz function in the second/first quoter of the plane, respectively. By general properties of Herglotz functions \cite[Problem 293]{Polya} we conclude that

\begin{corollary} \la{monotone} $ Q ( i \von ) $ is a selfadjoint monotone increasing function of $ \von > 0 $. \end{corollary}   

The monotonicity claimed by this corollary provides another, historically first, proof of finiteness of $ \sigma_+ ( T ) $ \cite{JLh,LW}. 

From now on $ \langle \cdot , \cdot \rangle $ stands for the inner product in $ L^2 ( \R ) $. The assertion of Proposition \ref{S0any} then holds with $ \ell = \frac 1{\sqrt 2} \sqrt c $. Let $ \left\{ \eta _n ( \von ) \right\} _{n=1}^\infty $, $ \eta_n ( \von ) \le \eta_{ n+1 } ( \von ) $, be the eigenvalues of the operator $ \tQ ( i\von ) $, and let ${\mathcal B}_c = \{ k_n = \lim_{\von \downarrow 0 }
\eta_n ( \von ) \} $. Applying the Nagy-Foias criterion and taking into account Corollary \ref{sup}, we find that \cite[Proposition 3.9 and Corollary 3.7]{KNR}

\begin{corollary}\la{sepEco}
The following are equivalent

${\mathrm (i)}$ $  0 $ is a spectral singularity,

${\mathrm (ii)}$ $T_{ess} $ is not similar to a self-adjoint operator,

${\mathrm (iii)}$
$ \ker S ( 0 ) \ne \{ 0 \} $,

${\mathrm (iv)}$ $ -1 \in {\mathcal B}_c $.
\end{corollary}

We say that $ c \in \cE $ if any of these equivalent conditions is satisfied. Note
that for any nonzero $c $ the function $\varkappa c $ belongs to $
\cE $ if $ - 1/\varkappa \in {\mathcal B}_c $, as was mentioned in Introduction.

Let $ Y $ and $ Y_1 $ be the integral operators in $ E $ with kernels $ \frac 12 \sqrt {c ( x )} ( \gamma + \ln | x - y | ) \sqrt {c ( y )} $ and $ \frac 12 \sqrt {c ( x )} | x - y | \sqrt {c ( y )} $, respectively, $ \mathrm{N} = \{ u \in \ker ( I + Y ) \colon \llangle u , \sqrt{c} \rrangle = 0 \} $, and $ P_\mathrm{N} $ be the orthogonal projection on $ \mathrm{N} $ in $ E $. Further analysis splits into two cases depending on whether the subspace $ \mathrm{N} $ is trivial or not. 

\begin{theorem}
Let $ c \in \cE $. If the subspace $ \mathrm{N} $ is trivial, then for all $ \xi > 0 $ save for at most one the number $ p( \xi ) \equiv \mbox{dist} \( -1,
\sigma( \tQ ( i\xi ))\) \neq 0 $, and for any such $ \xi $ the following asymptotics holds in the operator norm when $ z \to 0 $ in $ \overline{\C}_+ $, \be\label{Slog}  S^{ -1 } ( z ) = G - \ln \( \frac {-iz}\xi \) \llangle \cdot , \tilde{e} \rrangle \tilde{e} + O ( | z \ln^2 z| ) , \ee where \[ G = \frac{ I - \tQ ( i \xi ) }{ I + \tQ ( i \xi ) } , \] \[ \tilde{e} = \( I + \tQ ( i \xi ) \)^{ -1 } \sqrt c . \] If the subspace $ \mathrm{N} $ is non-trivial, then the operator $ M = \left. P_{\mathrm{N}} Y_1 \right|_{\mathrm{N}} $ in the space $ \mathrm{N} $ is invertible, and the following asymptotics holds in the operator norm when $ z \to 0 $ in $ \overline{\C}_+ $, provided that the restriction of $ I + Y $ to its reducing subspace $ \mathrm{N}^\perp $ is invertible, \begin{eqnarray} \label{Ssimple}  & & S^{ -1 } ( z ) = -\frac 1z 2i M^{ -1 } P_N + B_0 + \nonumber \\  & & \left\{ \begin{array}{cc}  0, & \vartheta \ne 0, \\ \noalign{\vskip5pt} - \ln ( iz ) \llangle \cdot , \Lambda^* \sqrt c \rrangle \( I - M^{ -1 } P_N Y_1 \) \Lambda \sqrt c , & \vartheta = 0, \end{array} \right. + O \( \left| \frac 1{\ln z} \right| \), \\ \noalign{\vskip3pt} & & B_0 = -I + 2 \( I - M^{ -1 } P_N Y_1 \) \Lambda P_{\mathrm{N}}^\perp \( I - Y_1 M^{ -1 } P_N \) , \nonumber \\ \noalign{\vskip3pt} & & 
\Lambda = \( \left. \( I + Y \) \right|_{ \mathrm{N}^\perp } \)^{ -1 } , \vartheta = \llangle \Lambda \sqrt c , \sqrt c \rrangle  . \nonumber \end{eqnarray} \end{theorem}

Notice that the first alternative in this theorem includes the case when $ \ker ( I + Y ) $ is trivial. It is easy to see \cite{KNR} that there are $ c $'s such that $ \mathrm{N}$ is non-trivial and hence the second alternative in the theorem is realized. For instance, this is the case for $ c  ( x ) = \varkappa \left\{ \begin{array}{cc} 1 ,& |x| < a \\ 0 , & |x| > a \end{array} \right.  $ with an appropriate constant $ \varkappa > 0 $. 

\begin{proof}  Let the subspace $ \mathrm{N} $ be trivial. Then $ p( \xi ) \equiv \mbox{dist} \( -1,
\sigma( \tQ ( i\xi ))\) \neq 0 $ for all $ \xi > 0 $ except for at most one, for otherwise a non-zero linear combination of two elements of $ \ker \( I + \tQ ( i \xi ) \) $ corresponding to two different $ \xi $ would belong to $ \mathrm{N} $. Fix a $ \xi $ such that $ p ( \xi ) \ne 0 $.
Let $\alpha (z) =
\frac 12 \ln \( - \xi^{ -1 } i z \) $. To establish (\ref{Slog}), we calculate the inverse of $ I + \tQ ( z ) $ to obtain, \[ \( I + \tQ (z) \)^{-1} = \( I
+ \tQ ( i\xi ) \)^{-1} - \frac { \alpha (z) } { 1 + \alpha (z)
\varrho_c } \llangle \cdot , \tilde{e} \rrangle \tilde{e},  \] where $ \varrho_c  =
\llangle \tilde{e} , \sqrt c \rrangle
$. Since $ c \in \cE $, the number $ \varrho_c  = 0 $, for otherwise the function $ \( I + \tQ ( z ) \)^{ -1 } $ would be bounded at $ z = 0 $. Expressing $ \( I + Q ( z ) \)^{ -1 } $ via $ \( I + \tQ ( z ) \)^{ -1 } $ through the resolvent identity and taking into account that $ \Theta ( z ) = O ( z ) $ when $ z \to 0 $, we infer 
\[ \( I + Q (z) \)^{-1} = \( I 
+ \tQ ( i\xi ) \)^{-1} - \alpha (z) \llangle \cdot , \tilde{e} \rrangle \tilde{e} + O \( | z \ln^2 z | \) , \; z \in \omede ( 0 ) , \] and (\ref{Slog}) follows from (\ref{Char}) rewritten as \be\label{minus1} S^{-1} (z) = -I + 2 \( I + Q (z) \)^{-1} . \ee 

Now, let the subspace $ \mathrm{N} $ be non-trivial. First, we shall show that $ M $ is invertible. Indeed, let $ h \in E $ be real and such that $ \llangle \sqrt c , h \rrangle = 0 $, and let $ f = \sqrt c h $. Then, on integrating by parts we find, 
\bequnan \llangle Y_1 h , h \rrangle = \int_{-a}^a \int_{ -a }^x ( x - y )  f ( x ) f ( y ) \mathrm{d} y\, \mathrm{d} x = \( \int_{ -a }^a f ( x ) \mathrm{d} x  \) \int_{ -a }^a ( a - y ) f ( y ) \mathrm{d} y  \\ - \int_{-a}^a \mathrm{d} x \left| \int_{ -a }^x  f ( y ) \mathrm{d} y \right|^2 . \eequnan The boundary term in the right hand side vanishes, because $ \int f ( x ) \mathrm{d} x = 0 $ by the orthogonality condition, to give \[ \llangle Y_1 f , f \rrangle = - \int_{-a}^a \mathrm{d} x \left| \int_{ -a }^x  f ( y ) \mathrm{d} y \right|^2 \] which vanishes if, and only if, $ f = 0 $. It follows that $ \ker M = \{ 0 \} $, for the quadratic form of $ Y_1 $, obviously, vanishes on any $ h \in \ker M $, and hence the operator $ M $ is invertible, since $ \mathrm{N} $ is finite dimensional. 

Let us establish (\ref{Ssimple}). Define $ \beta ( z ) = \( 1/2 \) \ln ( -iz ) $, $ Q_1 (z) = Y + \beta ( z) \llangle \cdot , \sqrt c \rrangle \sqrt c + iz Y_1 $. First, we calculate the asymptotics of $ \( I + Q_1 ( z ) \)^{ -1 } $. Consider the equation $ ( I + Q_1 ( z ) ) f = g $ in components, $ f_\mathrm{N} = P_\mathrm{N} f $,  $ f_\mathrm{N}^\perp = f - f_\mathrm{N} $,
\bequnan ( I + Y ) f_\mathrm{N}^\perp + i z P_\mathrm{N}^\perp Y_1 f + \beta ( z ) \llangle f_\mathrm{N}^\perp , \sqrt c \rrangle P_\mathrm{N}^\perp \sqrt c = P_\mathrm{N}^\perp g, \\ i z P_\mathrm{N} Y_1 f = P_\mathrm{N} g . \eequnan Using the invertibility of $ M $, we solve the second equation with respect to $ f_\mathrm{N} $, and plug the result into the first one to find out \bequnan  f_\mathrm{N} = \frac 1{iz} M^{ -1 } P_\mathrm{N} g - M^{ -1 } P_\mathrm{N} Y_1  f_\mathrm{N}^\perp , \\ V f_\mathrm{N}^\perp + \beta ( z ) \llangle f_\mathrm{N}^\perp , \sqrt c \rrangle \sqrt c = P_\mathrm{N}^\perp \( g - Y_1 M^{ -1 } P_\mathrm{N} g \) , \eequnan where $ V $ is an operator in $\mathrm{N}^\perp $ such that $ V = \left. ( I + Y ) \right|_{\mathrm{N}^\perp } + O ( |z| ) $.  Let $ \vartheta \ne 0 $. Solving the second equation taking into account that $ V $ is invertible and $ V^{ -1} = G + O ( |z| ) $, and substituting the result back, we obtain 
\bequnan f =   \frac 1{iz} M^{ -1 } P_N g + \( I - M^{ -1 } P_N Y_1 \) \Lambda P_{\mathrm{N}}^\perp \( I - Y_1 M^{ -1 } P_N \)g - \\ \frac 1\vartheta  \llangle g , \Lambda^* \sqrt c \rrangle \( I - M^{ -1 } P_N Y_1 \) \Lambda \sqrt c + O \( \left| \frac 1{\ln z} \right| \) \| g \|. \eequnan Now, that the asymptotics of $ f = \( I + Q_1 ( z ) \)^{ -1 } g $ is found, we notice that $ Q ( z ) - Q_1 ( z ) = O \( \left| z \right|^2 \)  $, and therefore the asymptotics of $ \( I + Q ( z ) \)^{ -1 }g $ is the same. The asymptotics (\ref{Ssimple}) follows from (\ref{minus1}). The case $ \vartheta = 0 $ is treated similarly.
\end{proof}

The assumption of invertibility of $ I + Y $ on the subspace $ \mathrm{N}^\perp $ in the second part of the theorem is only made for convenience. If the restriction of $ I + Y $ to $ \mathrm{N}^\perp $ is not invertible, then the restriction of $ I + \tilde Y $, $ \tilde Y  = Y + \frac 12 \ln \delta \llangle \cdot , \sqrt c \rrangle \sqrt c $ with a $ \delta \ne 1 $, is. The asymptotics (\ref{Ssimple}) remains valid if we substitute $ \tilde Y $ for $ Y $, $ z / \delta $ for $ z $, and $ \delta Y_1 $ for $ Y_1 $ in there.

\begin{corollary}
The characteristic function satisfies \be \la{firstorderest} \len S^{ -1 } ( z ) \rin
\le \frac C{|z|} \ee for $ z \in \omede ( 0 ) $. If the subspace $ \mathrm{N} $ is non-trivial, then $ 0 $ is the spectral singularity of the first order in the strict sense, \be\la{optimal} \len S ( k ) h \rin = O ( |k| ) , \; k \to 0  \ee for a non-zero $ h \in E $ (in fact, for any $ h \in \mathrm{N} $).
\end{corollary}

This corollary is immediate from asymptotics (\ref{Slog}) and (\ref{Ssimple}).
One can now apply results from \cite{NR}. Theorem 3 in that paper states that 

{\it If a maximal dissipative operator $ D $ has a spectral singularity of order $ p > 0 $ in the strict sense then for any sufficiently small $
\von > 0 $ there exists a $ u \in H_{ac} ( D ) $ such that
\[ \len e^{ -itD } u \rin = t^{ 1 - \von } ( 1 + o
( 1 ) ), \;\;\; t \to + \infty . \] }

An assertion similar to this theorem and \cite[Corollary 2]{NR} cited at the end of section 3, holds for logarithmic singularities \cite[Corollary 4 and Theorem 5]{KNR}. Combining these and taking into account the unitary equivalence between $ L $ and $ -T $, we arrive at the following conclusion.

\begin{corollary}\la{semigroest} If $ c \in \cE $ then \[ \len \left. e^{ itL } \right|_{ J H_{ess} } \rin \le C ( 1 + t ) \] for some $ c \ne 0 $. If $ \mathrm{N} $ is non-trivial, then for any sufficiently small $
\von > 0 $ there exists a $ u \in J H_{ ess } $ such that
\be\la{exactness} \len e^{ itL } u \rin = t^{ 1 - \von } ( 1 + o
( 1 ) ), \;\;\;  t \to + \infty . \ee If $ \mathrm{N} $ is trivial, then $$ \len \left. e^{ itL }  \right|_{ J H_{ess}}
\rin \le C \ln t $$ for all $ t \ge 2 $, and for any
sufficiently small $ \von > 0 $ there exists a $ u \in J H_{ess} $
such that $$ \len e^{ itL } u \rin = \( \ln t \)^{ 1 - \von } ( 1
+ o ( 1 ) ), \;\;\; t \to + \infty . $$ \end{corollary}

It is worth mentioning that in the three-dimensional problem with compactly supported $ c $ the characteristic function is analytic on the real axis \cite{KNR}, and has a simple zero at $ z = 0 $ when $ c \in \cE $. The appearance of the logarithmic case in Theorem 2 is specific for the slab problem.

\begin{remark} It is to be expected that there exist polynomial collision integrals $ K $ such that the operator has non-zero spectral singularities. Indeed, the reason for their absence in the isotropic case given by lemma \ref{propertiesQ} is no longer in force in the anisotropic case. \end{remark} 



\begin{thebibliography}{99}

\bibitem{KNR} Yu. Kuperin, S. Naboko \and R. Romanov.
Spectral analysis of the transport operator: a functional model
approach, \emph{Indiana Univ. Math. J.} \textbf{51} (2002), No. 6,
1389 - 1425.

\bibitem{JLh} J. Lehner. The spectrum of the neutron transport
operator for the infinite slab, \emph{J. Math. Mech.} \textbf{11}
(1962), No.  2, 173-181.

\bibitem{LW} J. Lehner \and G. Wing. On the spectrum of an
unsymmetric operator arising in the transport theory of neutrons,
\emph{Comm. Pure Appl. Math.} \textbf{8} (1955), 217--234.

\bibitem{LW2} J. Lehner \and G. Wing. Solution of the linearized Boltzmann
equation for the slab geometry, \emph{Duke Math. J.} \textbf{23}
(1956), 125--142.

\bibitem{N} S. N. Naboko. A functional model of perturbation theory and
its applications to scattering theory, \emph{Trudy MIAN} \textbf{147}
(1980), 86--114 (English transl.: {\it Proc.
Steklov Inst. Math.} (1981), No. 2, 85--116).

\bibitem{Nab} S. N. Naboko. On the conditions for existence of wave
operators in the nonselfadjoint case, \emph{Probl. Mat. Fiz.,} vol. 12
(M. Birman ed.), LGU, Leningrad, 1987, 132--155 (
English transl.: {\it Amer. Math. Soc. Transl.} (2) \textbf{157}
(1993), 127--149).

\bibitem{NR} S. Naboko \and R. Romanov. Spectral singularities and asymptotics of contractive semigroups. I, \emph{Acta Sci. Math. (Szeged)} \textbf{70} (2004), 379--403.

\bibitem{Nai} M.A.Na\u{i}mark. Investigation of the spectrum and
the expansion in eigenfunctions of a nonselfadjoint operator of
the second order on a semi-axis, \emph{Trudy Mosk. Mat. Obsch.}
\textbf{3} (1954), 181--270 (English transl.:
\emph{Amer. Math. Soc. Transl.} (2) \textbf{16} (1960), 103--193).

\bibitem{Nik} N. K. Nikolski\u\i . \emph{Treatise on the shift operator.
Spectral function theory}. (Springer-Verlag, Berlin, 1986).

\bibitem{Pavli}  B. S. Pavlov. On expansions in eigenfunctions of
the absolutely continuous spectrum of a dissipative operator, \emph{
Vestnik Leningrad Univ. (Mat. Mech. Astronom)} (1975) vyp. 1,
130--137 (English transl.: \emph{Vestnik
Leningrad Univ. Math.} \textbf{8} (1980), 135--143).

\bibitem{Pav} B. S. Pavlov. On separation conditions for the
spectral components of a dissipative operator, \emph{Math. USSR
Izvestija} \textbf{9} (1975), No. 1, 113--137.

\bibitem{PavEn} B. S. Pavlov. Spectral analysis of a dissipative
singular Schrodinger operator in terms of a functional model, in
 \emph{Encyclopaedia of Mathematical Sciences} \textbf{65} (M. A. Shubin ed.);
(Springer-Verlag, Berlin, 1996), pp. 87--153.

\bibitem{Polya} G. Polya \and G. Szeg\H{o}. \emph{Aufgaben und Lehrs\"{a}tze aus der Analysis. I}. 
(Springer-Verlag, Berlin, 1970). 

\bibitem{equival} R. Romanov. A remark on equivalence of weak and strong definitions
of the absolutely continuous subspace for nonself-adjoint
operators, in \emph{Spectral Methods for Operators of Mathematical Physics (Bedlewo 2002)}, Oper. Theory: Adv. Appl. \textbf{154}, pp. 179--184 (Birkh\"{a}user, Basel, 2004).

\bibitem{Sachn} L.A.Sahnovi\v{c}. Dissipative operators with
absolutely continuous spectrum, \emph{Trans. Moscow Math. Soc.}
\textbf{19} (1968).

\bibitem{Na} B. Sz.-Nagy \and C.Foia\c{s}. \emph{Analyse Harmonique des
Operateurs de l$^{\prime }$Espase de Hilbert}. (Masson et C$^{ie}$/
Academiai Kiado, 1967).

\bibitem{Tikhang} A. S. Tikhonov.
Functional model and duality of spectral components for operators
with continuous spectrum on a curve, \emph{Algebra i Analiz} \textbf
{14} (2002), no. 4,  158--195 (English transl.:
\emph{St. Petersburg Math. J.} \textbf{14}:4 (2003), pp. 655--682).

\end{thebibliography}
\end{document}